\documentclass{birkjour}
\newtheorem{thm}{Theorem}[section]

\newtheorem{lem}[thm]{Lemma}

\numberwithin{equation}{section}

\begin{document}

%
%
%
%
%
%
%
%
%
\title[Logarithmic uncertainty principle for Clifford wavelet transform]{Logarithmic uncertainty principle for Clifford-wavelet transform}
\author[Arfaoui]{Sabrine Arfaoui}
\address{Sabrine Arfaoui, $^1$Laboratory of Algebra, Number Theory and Nonlinear Analysis LR18ES15, Department of Mathematics, Faculty of Sciences, University of Monastir, Monastir 5019, Tunisia.\\ $^2$Higher Institute of Computer Science El-Kef, University of Jendouba, 5 Street of Saleh Ayech - Kef 7100, Tunisia.}
\email{sabrine.arfaoui@issatm.rnu.tn}
\author[Banouh]{Hicham Banouh}
\address{Departement of Mathematics, Faculty of Sciences and Applied Sciences, University of Akli Mohand Oulhadj of Bouira, Bouira 10000, Algeria.}
\email{h.banouh@univ-bouira.dz}
\author[Ben Mabrouk]{Anouar Ben Mabrouk}
\address{$^1$Higher Institute of Applied Mathematics and Computer Science, University of Kairaouan, Street of Assad Ibn Alfourat, Kairaouan 3100, Tunisia.\\
$^2$Laboratory of Algebra, Number Theory and Nonlinear Analysis LR18ES15, Department of Mathematics, Faculty of Sciences, Monastir 5019, Tunisia.\\
$^3$Department of Mathematics, Faculty of Sciences, University of Tabuk, King Faisal Rd, Tabuk 47512, Saudi Arabia.}
\email{anouar.benmabrouk@fsm.rnu.tn}


\begin{abstract} 
The uncertainty principle constitutes one of the famous physical concepts which continues to attract researchers from different related fields since its discovery due to its utility in many applications. Among the classical (Fourier-based) forms, modified variants such as wavelet-based and logarithmic versions have been investigated by researchers. In the present work we are concerned with the development of a new logarithmic uncertainty principle combining both wavelets and logarithmic variants in the Clifford wavelet transform framework.
\end{abstract}
\maketitle
\section{Introduction}
In the last decade(s) a coming back to Clifford algebras and analysis has taken place in many fields such as mathematics, physics and computer science. This may be due to the fact that Clifford algebras incorporate inside one single structure the geometrical and algebraic properties of Euclidean space. Incite of a Clifford algebra geometric entities are treated according to their dimension such as scalars, vectors, bivectors and volume elements. This allowed researchers especially in mathematical physics to develop harmonic analysis on Clifford algebras as extension of complex one. Famous operators and related inequalities
such as Cauchy-Riemann operator and uncertainty principle have been extended to the Clifford framework.

Recently extensions of wavelet theory to the Clifford framework have been developed and applied in many fields, and consequently connection between wavelet theory and Clifford algebra/analysis is now growing up, and applications are widespread such as in signal/image processing in physics, nuclear magnetics, electric engineering, colour images.
See \cite{DeBie2,Hitzer2,Hitzer-Mawardi-1,Mawardi-Hitzer-1}.

Wavelet analysis of functions starts by checking that the mother wavelet is admissible to guarantee the reconstruction, Fourier-Plancherel and Parseval rules.

The present work is concerned with the development of a sharp uncertainty principle in the Clifford framework based on Clifford wavelets and thus extending the one developed in \cite{Banouh2,Banouh4}. We propose a new Heisenberg uncertainty principle applied to continuous Clifford wavelet transform in the settings of the non commutative Clifford algebras. In section 2 some preliminaries on Clifford toolkit and
wavelets are presented. We also recall briefly the most recent existing works on the uncertainty principle in both its classical forms and Clifford one. Section 4 is devoted to the development of our main results. 
\section{Clifford wavelets toolkit}
In this section, we review briefly the basic concepts of Clifford analysis and the continuous wavelet transform on the real Clifford algebra as well as the uncertainty principle .

The Clifford algebra $\mathbb{R}_{n}$ associated to $\mathbb{R}^{n}$ is the $2^{n}$-dimensional non commutative associative algebra generated by the canonical basis $\left\{ e_{1},e_{2},\dots,e_{n}\right\}$ relatively to the non commutative product
$$
e_{i}e_{j}+e_{j}e_{i}=-2\delta_{ij},
$$
where $\delta$ is the Kronecker symbol. Any element $a\in\mathbb{R}_{n}$ has a unique representation 
\[
a=\sum_{A}a_{A}e_{A}=\sum_{k=0}^{n}\sum_{\left|A\right|=k}a_{A}e_{A},\;a_{A}\in\mathbb{R},
\]
where $|A|$ is the length of the multi-index $A$. We define here the module (or the norm) of $a$ in the Clifford algebra $\mathbb{R}_n$ by
\[
|a|=(\sum_{A}|a_{A}|^2)^{1/2},
\]
where for $x\in\mathbb{R}$, $|x|$ stands for its absolute value. 

On the algebra $\mathbb{R}_{n}$ there are some types of involutive operators defined on the basis elements. A main-involution satisfying
$$
\widetilde{e_{A}}=(-1)^{\left|A\right|}e_{A}.
$$
A second type called reversion stating that 
$$
e_{A}^{\ast}=(-1)^{\frac{\left|A\right|(\left|A\right|-1)}{2}}e_{A}
$$
and finally a conjugation operator defined by 
$$
\overline{e_{A}}=(-1)^{\frac{|A|(|A|+1)}{2}}e_{A}.
$$
In the case of the complex extension $\mathbb{C}_{n}=\mathbb{R}_{n}+i\mathbb{R}_{n}$
the last operator has to be extended to the so-called Hermitian conjugation
defined by 
$$
\lambda^{\dagger}=\overline{a}-i\overline{b},
$$
for $\lambda=a+ib\in\mathbb{C}_{n}$ with $a,b\in\mathbb{R}_{n}$. The module (or the norm) of $\lambda$ in the Clifford algebra, will be 
$$
|\lambda|^2=\lambda^{\dagger}\lambda=|a|^2+|b|^2,
$$
As $\mathbb{R}^{n}$ is a subspace of $\mathbb{R}_{n}$ any element $x=(x_{1},x_{2},\dots,x_{n})\in\mathbb{R}^{n}$ is identified to $\underline{x}={\displaystyle \sum_{j=1}^{n}e_{j}x_{j}}$. This allows to introduce the Clifford product of vectors by 
$$
\underline{x}\underline{y}=-<\underline{x},\underline{y}>+\underline{x}\wedge\underline{y},
$$
where $<\underline{x},\underline{y}>$ is the usual inner product and $\underline{x}\wedge\underline{y}$ is the exterior product on $\mathbb{R}^{n}$. 

Finally we recall that for two functions $f,g:\mathbb{R}^{n}\longrightarrow\mathbb{R}_{n}$ (the same apply for functions with values in $\mathbb{C}_{n}$) the
inner product is defined by 
\begin{equation}
<f,g>_{L^{2}(\mathbb{R}^{n},\mathbb{R}_{n},dV(\underline{x}))}=\int_{\mathbb{R}^{n}}\left[f(\underline{x})\right]^{\dagger}g(\underline{x})dV(\underline{x})\label{eq:inner product}
\end{equation}
and the associated norm by 
\[
\|f\|_{L^{2}(\mathbb{R}^{n},\mathbb{R}_{n},dV(\underline{x}))}=<f,f>_{L^{2}(\mathbb{R}^{n},\mathbb{R}_{n},dV(\underline{x}))}^{\frac{1}{2}}.
\]
We denote also 
\[
\left\Vert f\right\Vert _{L^{1}(\mathbb{R}^{n},\mathbb{R}_{n},dV(\underline{x}))}=\int_{\mathbb{R}^{n}}|f(\underline{x})|dV(\underline{x})
\]
where $dV(\underline{x})$ stands for the Lebesgue measure, and $|f(x)|^2=\left[f(\underline{x})\right]^{\dagger}f(\underline{x})$

As for multidimensional Euclidean spaces the Fourier and wavelet transforms have been extended to the case of Clifford framework. The Clifford-Fourier transform of a Clifford-valued function $f\in L^{1}\cap L^{2}(\mathbb{R}^{n},\mathbb{R}_{n},dV(\underline{x}))$
is defined by 
\[
\mathcal{F}\left[f\right](\underline{\xi})=\widehat{f}(\underline{\xi})=\frac{1}{(2\pi)^{\frac{n}{2}}}\int_{\mathbb{R}^{n}}e^{-i<\underline{x},\underline{\xi}>}f(\underline{x})dV(\underline{x}).
\]
In the present paper we use spin theory as in \cite{Banouh1,Banouh2,Banouh4}.
Let $\psi\in L^{1}\cap L^{2}(\mathbb{R}^{n},dV(\underline{x}))$ be
such that 
\begin{itemize}
\item ${\psi}$ is a Clifford-valued function. 
\item $\widehat{\psi}(\underline{\xi})\left[\widehat{\psi}(\underline{\xi})\right]^{\dagger}$ is scalar. 
\item $\mathcal{A}_{\psi}={\displaystyle (2\pi)^{n}\int_{\mathbb{R}^{n}}\frac{\widehat{\psi}(\underline{\xi})\left[\widehat{\psi}(\underline{\xi})\right]^{\dagger}}{|\underline{\xi}|^{n}}dV(\underline{\xi})<\infty.}$ 
\end{itemize}
$\psi$ is called a Clifford mother wavelet. The spin group of order $n$ is defined by 
\[
S_{n}=\left\{ s\in\mathbb{R}_{n};\;s={\displaystyle \prod\limits _{j=1}^{2l}}\underline{\omega}_{j},\;\underline{\omega}_{j}^{2}=-1,1\leq j\leq2l\right\} .
\]
For $(a,\underline{b},s)\in\mathbb{R}_{+}\times\mathbb{R}^{n}\times S_{n}$, we denote 
\[
\psi^{a,\underline{b},s}(\underline{x})=\frac{1}{a^{\frac{n}{2}}}s\psi(\frac{\overline{s}(\underline{x}-\underline{b})s}{a})\overline{s}.
\]
In \cite{Banouh1,Banouh2,Banouh4} the authors showed that the copies ${\psi^{a,\underline{b},s}}$ constitute a dense set in $L^{2}(\mathbb{R}^{n},\mathbb{R}^{n},dV(\underline{x}))$.
The Clifford-wavelet transform of an analyzed function $f\in L^{2}(\mathbb{R}^{n},\mathbb{R}^{n},dV(\underline{x}))$
is defined by 
\[
T_{\psi}\left[f\right](a,\underline{b},s)=\int_{\mathbb{R}^{n}}\left[\psi^{a,\underline{b},s}(\underline{x})\right]^{\dagger}f(\underline{x})dV(\underline{x}).
\]
Let $\mathcal{H}_{\psi}=T_{\psi}\left(L^{2}(\mathbb{R}^{n},dV(\underline{x}))\right)$. We define an inner product on $\mathcal{H}_{\psi}$ by 
\[
\left[T_{\psi}\left[f\right],T_{\psi}\left[g\right]\right]=\frac{1}{\mathcal{A}_{\psi}}\int\limits _{S_{n}}\int\limits _{\mathbb{R}^{n}}\int\limits _{\mathbb{R}_{+}}\mathcal{K}_{\psi}\left[f,g\right](a,\underline{b},s)d\mu(a,\underline{b},s),
\]
where 
\[
\mathcal{K}_{\psi}\left[f,g\right](a,\underline{b},s)=(T_{\psi}\left[f\right](a,\underline{b},s))^{\dagger}T_{\psi}\left[g\right](a,\underline{b},s),
\]
$$
d\mu(a,\underline{b},s)=\frac{da}{a^{n+1}}dV(\underline{b})ds,
$$
and where $ds$ is the Haar measure on $S_{n}$.

It is straightforward that the operator $T_{\psi}$ is an isometry from $L^{2}(\mathbb{R}^{n},dV(\underline{x}))$ to $\mathcal{H}_{\psi}$ ,and we have 
\begin{equation}
\int\limits _{Spin(n)}\int\limits _{\mathbb{R}^{n}}\int\limits _{\mathbb{R}^{+}}\left|T_{\psi}\left[f\right](a,\underline{b},s)\right|{}^{2}\frac{da}{a^{n+1}}dV(\underline{b})ds=\mathcal{A}_{\psi}\left\Vert f\right\Vert _{2}^{2}.\label{eq:norme equality}
\end{equation}
See \cite{Banouh1,Banouh2,Banouh4}. As a result any analyzed function in $L^{2}(\mathbb{R}^{n},dV(\underline{x}))$ may be reconstructed in the $L^{2}$-sense by means of its Clifford-wavelet transform which constitutes the Clifford-wavelet Plancherel-Parseval formulas.

Our purpose in the present work is to establish a logarithmic uncertainty principle in the Clifford framework by using the Clifford wavelet continuous transform.

The uncertainty principle originally due to Heisenberg constitutes since its discovery an interesting concept in quantum mechanics. Mathematically, it is resumed by the fact that a non-zero function and its Fourier transform cannot both be sharply localized. In the last decades many studies have been developed on the uncertainty principle using pure
Fourier transform, pure wavelets, Clifford Fourier transform and recently Clifford wavelet transform. We will not review here in details these works. However the readers may refer to \cite{Rachdi-Meherzi,Rachdi-Herch,ElHaoui-Fahlaoui,Hitzer2,Hitzer-Mawardi-1,Mawardi-Ryuichi-1,Mawardi-Ryuichi-2,Mawardi-Hitzer-1,Mawardi-Hitzer-2,Mawardi-Hitzer-3,Mawardi-Hitzer-4,Mawardi-Hitzer-Hayashi-Ashino,Kouetal}.
\section{Logarithmic uncertainty for Clifford Wavelet Transform}
Based on the classical Pitt’s inequality, Beckner \cite{Beckner1995} proved the logarithmic version of Heisenberg’s uncertainty principle. Recently, this principle has been carried out for different two-dimensional time-frequency domain transforms [2,7,19]. Here, we derive the logarithmic inequality in Clifford wavelet transform domains \cite{Brahim2019,Brahim2019a,Chen2015,Gorbachev2016,Beckner1995,Folland1997,Lian2020,Mawardi-Ryuichi-2,ElHaouietal1,ElHaouietal}. 

In \cite{Banouh1,Banouh2,Banouh3,Banouh4}, the authors have established the two following results on the uncertainty principle applied to the continuous Clifford wavelet transform.
\begin{thm}\cite{Banouh1,Banouh2,Banouh3,Banouh4}\label{MainTheorem}
Let $\psi\in L^{2}(\mathbb{R}_{n},dV(\underline{x}))$ be an admissible Clifford mother wavelet. Then, for $f\in L^{2}(\mathbb{R}_{n},dV(\underline{x}))$ the following inequality holds
\begin{equation}
\displaystyle\left(\displaystyle\int_{Spin(n)}\displaystyle\int_{\mathbb{R}^{+}}\left\Vert b_{k}T\left[f\right](a,\cdot,s)\right\Vert _{2}^{2}\frac{da}{a^{n+1}}ds\right)^{\frac{1}{2}}\left\Vert \xi_{k}\widehat{f}\right\Vert _{2}\geq\frac{(2\pi)^{\frac{n}{2}}}{2}\sqrt{A_{\psi}}\left\Vert f\right\Vert _{2}^{2},
\end{equation}
where $k=1,2,\cdots,n$.
\end{thm}
\begin{thm}\cite{Banouh2,Banouh3,Banouh4}\label{new cliff uncer} Let $\psi\in L^{2}(\mathbb{R}^{n},\mathbb{R}_{n},dV(\underline{x}))$ be an admissible Clifford mother wavelet. Then for $f\in L^{2}(\mathbb{R}^{n},\mathbb{R}_{n},dV(\underline{x}))$ 	the following inequality holds
\small
\[
(\int_{Spin(n)}\int_{\mathbb{R}^{+}}\left\Vert b_{k}T_{\psi}\left[f\right](a,\bullet,s)\right\Vert_{2} ^{2}\frac{da}{a^{n+1}}ds)^{\frac{1}{2}}\left\Vert \xi_{k}\widehat{f}\right\Vert_{2} \geq\sqrt{2^{n+1}\pi{}^{n}A_{\psi}}\left\{\left\Vert f\right\Vert_{2} ^{2}+2\left|\left\langle f_1,f_2\right\rangle \right|\right\}
\]
\normalsize
where
\[
\begin{cases}
f_1(\underline{x}) & =\dfrac{1}{A_{\psi}}\displaystyle\int_{Spin(n)}\displaystyle\int_{\mathbb{R}^{n}}\int_{\mathbb{R}^{+}}\psi^{a,\underline{b},s}(\underline{x})\partial_{b_{k}}T_{\psi}\left[f\right](a,\underline{b},s)\frac{da}{a^{n+1}}dV(\underline{b})ds\\
f_2(\underline{x}) & =\dfrac{1}{A_{\psi}}\displaystyle\int_{Spin(n)}\int_{\mathbb{R}^{n}}\int_{\mathbb{R}^{+}}\psi^{a,\underline{b},s}(\underline{x})b_{k}T_{\psi}\left[f\right](a,\underline{b},s)\frac{da}{a^{n+1}}dV(\underline{b})ds
\end{cases}
\]
\end{thm}
Now we present the main result of this work. First we give 
\begin{thm}\cite{Gorbachev2016,Lian2020,Beckner1995,Folland1997}\label{uncer princ log fourier} For $f\in\mathcal{S}(\mathbb{R}^n,\mathbb{R}_n,dV(\underline{x}))$
\[
\int_{\mathbb{R}^{n}}\ln\left|\underline{x}\right|\left|f(\underline{x})\right|^{2}dV(\underline{x})+\int_{\mathbb{R}^{n}}\ln\left|\underline{\xi}\right|\left|\widehat{f}(\underline{\xi})\right|^{2}dV(\underline{\xi})\geq (\varphi(\frac{n}{4})+\ln(2))\left\Vert f\right\Vert _{2}^{2}
\]
where $\varphi(t)=\frac{d}{dt}\ln\Gamma(t)=\frac{\Gamma'(t)}{\Gamma(t)}$ is the Digamma function.
\end{thm}
\begin{thm}\label{loguncercliffwav} Let $\psi\in L^{2}(\mathbb{R}^{n},\mathbb{R}_{n},dV(\underline{x}))$ be an admissible Clifford mother wavelet. Then for $f\in\mathcal{S}(\mathbb{R}^{n},\mathbb{R}_{n},dV(\underline{x}))$
\[
\int_{Spin(n)}\int_{\mathbb{R}^{+}}\int_{\mathbb{R}^{n}}\ln\left|\underline{b}\right|\left|T_{\psi}\left[f\right](a,\underline{b},s)\right|^{2}dV(\underline{b})\frac{da}{a^{n+1}}ds	\]
\[
+\frac{A_{\psi}}{(2\pi)^{n}}\int_{\mathbb{R}^{n}}\ln\left|\underline{\xi}\right|\left|\widehat{f}(\underline{\xi})\right|^{2}dV(\underline{\xi})\geq(\frac{\Gamma'(\frac{n}{4})}{\Gamma(\frac{n}{4})}+\ln(2))A_{\psi}\left\Vert f\right\Vert_{2}^{2}.
\]
\end{thm}
To prove this result we need the following lemma.
\begin{lem}
\[
\int_{Spin(n)}\int_{\mathbb{R}^{+}}\int_{\mathbb{R}^{n}}\ln\left|\underline{\xi}\right|\left|\widehat{T_{\psi}\left[f\right]}(a,\underline{\xi},s)\right|^{2}dV(\underline{\xi})\frac{da}{a^{n+1}}ds
\]
\[
=\frac{A_{\psi}}{(2\pi)^{n}}\int_{\mathbb{R}^{n}}\ln\left|\underline{\xi}\right|\left|\widehat{f}(\underline{\xi})\right|^{2}dV(\underline{\xi})
\]
\end{lem}
\begin{proof} We have
\[
\int_{Spin(n)}\int_{\mathbb{R}^{+}}\int_{\mathbb{R}^{n}}\ln\left|\underline{\xi}\right|\left|\widehat{T_{\psi}\left[f\right]}(a,\underline{\xi},s)\right|^{2}dV(\underline{\xi})\frac{da}{a^{n+1}}ds
\]
\[
=\int_{Spin(n)}\int_{\mathbb{R}^{+}}\int_{\mathbb{R}^{n}}\ln\left|\underline{\xi}\right|\left|a^{\frac{n}{2}}\overline{s}\left[\widehat{\psi}(a\overline{s}\underline{\xi}s)\right]^{\dagger}s\widehat{f}(\underline{\xi})\right|^{2}dV(\underline{\xi})\frac{da}{a^{n+1}}ds
\]
\[
=\int_{Spin(n)}\int_{\mathbb{R}^{+}}\int_{\mathbb{R}^{n}}\ln\left|\underline{\xi}\right|a^{\frac{n}{2}}\overline{s}\left[\widehat{\psi}(a\overline{s}\underline{\xi}s)\right]^{\dagger}s\widehat{f}(\underline{\xi})
\]
\[
\left[a^{\frac{n}{2}}\overline{s}\left[\widehat{\psi}(a\overline{s}\underline{\xi}s)\right]^{\dagger}s\widehat{f}(\underline{\xi})\right]^{\dagger}dV(\underline{\xi})\frac{da}{a^{n+1}}ds
\]
\[
=\int_{Spin(n)}\int_{\mathbb{R}^{+}}\int_{\mathbb{R}^{n}}\ln\left|\underline{\xi}\right|a^{\frac{n}{2}}\overline{s}\left[\widehat{\psi}(a\overline{s}\underline{\xi}s)\right]^{\dagger}
\]
\[
s\widehat{f}(\underline{\xi})\left[\widehat{f}(\underline{\xi})\right]^{\dagger}\overline{s}\widehat{\psi}(a\overline{s}\underline{\xi}s)sa^{\frac{n}{2}}dV(\underline{\xi})\frac{da}{a^{n+1}}ds.
\]
As $s\overline{s}=1$, we obtain
\[
\int_{Spin(n)}\int_{\mathbb{R}^{+}}\int_{\mathbb{R}^{n}}\ln\left|\underline{\xi}\right|\left|\widehat{T_{\psi}\left[f\right]}(a,\underline{\xi},s)\right|^{2}dV(\underline{\xi})\frac{da}{a^{n+1}}ds
\]
\[
=\int_{Spin(n)}\int_{\mathbb{R}^{+}}\int_{\mathbb{R}^{n}}\ln\left|\underline{\xi}\right|\left[\widehat{\psi}(a\overline{s}\underline{\xi}s)\right]^{\dagger}\widehat{f}(\underline{\xi})\left[\widehat{f}(\underline{\xi})\right]^{\dagger}\widehat{\psi}(a\overline{s}\underline{\xi}s)dV(\underline{\xi})\frac{da}{a}ds.
\]
we know from \cite[equation 3.5]{Banouh2,Banouh4} that $$\int_{Spin(n)}\int_{\mathbb{R}^{+}}\left[\widehat{\psi}(a\overline{s}\underline{\xi}s)\right]^{\dagger}\widehat{\psi}(a\overline{s}\underline{\xi}s)\frac{da}{a}ds=\frac{A_{\psi}}{(2\pi)^{n}}.
$$ 
So, we get
\begin{equation}
\int_{Spin(n)}\int_{\mathbb{R}^{+}}\int_{\mathbb{R}^{n}}\ln\left|\underline{\xi}\right|\left|\widehat{T_{\psi}\left[f\right]}(a,\underline{\xi},s)\right|^{2}dV(\underline{\xi})\frac{da}{a^{n+1}}ds
\end{equation}
\[
=\frac{A_{\psi}}{(2\pi)^{n}}\int_{\mathbb{R}^{n}}\ln\left|\underline{\xi}\right|\left|\widehat{f}(\underline{\xi})\right|^{2}dV(\underline{\xi}\label{lemma}.
\]
\end{proof}
\begin{proof}[Proof of theorem \ref{loguncercliffwav}] First we apply theorem \ref{uncer princ log fourier} to the function $\underline{b}\longmapsto T_{\psi}\left[f\right](a,\underline{b},s)$. So we obtain
\[
\int_{\mathbb{R}^{n}}\ln\left|\underline{b}\right|\left|T_{\psi}\left[f\right](a,\underline{b},s)\right|^{2}dV(\underline{b})+\int_{\mathbb{R}^{n}}\ln\left|\underline{\xi}\right|\left|\widehat{T_{\psi}\left[f\right]}(a,\underline{\xi},s)\right|^{2}dV(\underline{\xi})
\]
\[
\geq(\frac{\Gamma'(\frac{n}{4})}{\Gamma(\frac{n}{4})}+\ln(2))\left\Vert T_{\psi}\left[f\right]\right\Vert _{2}^{2}.
\]
Integrating by respect to the measure $\frac{da}{a^{n+1}}ds$ and taking in account Lemma \ref{lemma} and Formula (\ref{eq:norme equality}) we obtain
$$\int_{Spin(n)}\int_{\mathbb{R}^{+}}\int_{\mathbb{R}^{n}}\ln\left|\underline{b}\right|\left|T_{\psi}\left[f\right](a,\underline{b},s)\right|^{2}dV(\underline{b})\frac{da}{a^{n+1}}ds$$
$$+\frac{A_{\psi}}{(2\pi)^{n}}\int_{\mathbb{R}^{n}}\ln\left|\underline{\xi}\right|\left|\widehat{f}(\underline{\xi})\right|^{2}dV(\underline{\xi})\geq(\frac{\Gamma'(\frac{n}{4})}{\Gamma(\frac{n}{4})}+\ln(2))A_{\psi}\left\Vert f\right\Vert _{2}^{2}$$
\end{proof}
%

\section{Conclusion}
In this paper, a logarithmic uncertainty principle based on the continuous wavelet transform in the Clifford algebra's settings has been formulated and proved. Compared to existing formulations in Clifford framework, the proposed uncertainty principle here is sharper.
\section{Data availability statement}
Data sharing is not applicable to this article as no new data were created or analyzed in this study.
\section{Conflict of interest statement} All authors declare that they have no conflicts of interest regarding the present work.


\begin{thebibliography}{99}
\bibitem{Banouh1} H. Banouh, A. Ben Mabrouk, and M. Kesri, Clifford Wavelet Transform and the Uncertainty Principle. Adv. Appl. Clifford Algebras 29, 106 {2019}, 23 pages. 

\bibitem{Banouh2} H. Banouh, and A. Ben Mabrouk, Journal of Mathematical Physics, 61, 093502 (2020).

\bibitem{Banouh3} H. Banouh, Uncertainty principle associated with the continuous Clifford wavelet Transform. 12th International Conference on Clifford Algebras and Their Applications in Mathematical Physics (ICCA12). University of Science and Technology of China, Hefei, China, August 3 to 7, 2020.

\bibitem{Banouh4} H. Banouh, Analyse Harmonique par Ondelettes dans l’Algèbre de Clifford. Th\`ese de Doctorat en Math\'ematiques de l'Universit\'e des Sciences et de la Technologie Houari Boumediene, 2021.

\bibitem{Beckner1995} W. Beckner, Pitt's inequality and the uncertainty principle. Proceedings of the American Mathematical Society 123(6) (1995), pp. 1897—1897.

\bibitem{Brahim2019} K. Brahim, and E. Tefjeni, Uncertainty principle for the two-sided quaternion windowed Fourier transform. Integral Transforms and Special Functions 30(5) (2019), pp. 362--382.

\bibitem{Brahim2019a} K. Brahim, E. Tefjeni and B. Nefzi, Entropy and Logarithmic uncertainty principles for the multivariate continuous quaternion Shearlet Transform (2019). arXiv preprint arXiv:1912.08199.

\bibitem{Chen2015} L.-P. Chen, K. I. Kou, and M.-S. Liu, Pitt's Inequality and the Uncertainty Principle Associated with the Quaternion Fourier Transform. Journal of Mathematical Analysis and Applications, 423(1) (2015) pp. 681--700

\bibitem{Dahkleteal} S. Dahkle, G. Kutyniok, P. Maass, C. Sagiv, H.-G. Stark, and G. Teschke, The uncertainty principle associated with the continuous shearlet transform. International Journal of Wavelets, Multiresolution and Information Processing, 6(2), 157–181 (2008)

\bibitem{DeBie2} H. De Bie, and Y. Xu, On the Clifford-Fourier transform. International Mathematics Research Notices 2011, 5123 (2011), pp. 5123--5163. 

\bibitem{ElHaouietal1} Y. El Haoui, S. Fahlaoui, The continuous quaternion algebra-valued wavelet transform and the associated uncertainty principle. arXiv:1902.08461, (2019).

\bibitem{ElHaoui-Fahlaoui} Y. El Haoui, S. Fahlaoui, Donoho- Stark's Uncertainty Principles in Real Clifford Algebras. arXiv:1902.08465v1, 9 pages (2019).

\bibitem{ElHaouietal} Y. El Haoui, S. Fahlaoui, E. Hitzer, Generalized Uncertainty Principles associated with the Quaternionic Offset Linear Canonical Transform. arXiv:1807.04068v2 [math.CA], 15 pages (2019)

\bibitem{Feichtinger-Grochenig} H. G. Feichtinger, K. Gr\"ochenig, Gabor Wavelets and the Heisenberg Group: Gabor Expansions and Short Time Fourier Transform from the Group Theoretical Point of View. Wavelets, 359–397. doi:10.1016/b978-0-12-174590-5.50018-6 (1992)

\bibitem{Folland1997} G. B. Folland and A. Sitaram, The uncertainty principle: A mathematical survey. The Journal of Fourier Analysis and Applications 3(3) (1997), pp. 207--238.

\bibitem{Gorbachev2016} D. V. Gorbachev, V. I. Ivanov and S. Yu. Tikhonov, Sharp Pitt inequality and logarithmic uncertainty principle for Dunkl transform in $L^2$. Journal of Approximation Theory 202 (2016), pp. 109--118.

\bibitem{Hitzer2} E. Hitzer, Directional Uncertainty Principle for Quaternion Fourier Transform. Adv. appl. Clifford alg. 20, 271–284. DOI: 10.1007/s00006-009-0175-2 (2010)

\bibitem{Hitzer-Mawardi-1} E. Hitzer, and B. Mawardi, Uncertainty Principle for the Clifford-Geometric Algebra $\mathcal{Cl}_{3,0}$ based on Clifford Fourier Transform. arXiv:1306.2089v1 [math.RA] 4 pages (2013)

\bibitem{Kouetal} K. I. Kou, J.-Y. Ou, J. Morais, J.: On Uncertainty Principle for Quaternionic Linear Canonical Transform. Abstract and Applied Analysis, Article ID 725952, 14 pages. DOI: 10.1155/2013/725952 (2013)

\bibitem{Lian2020} P. Lian, Sharp inequalities for geometric Fourier transform and associated ambiguity function. Journal of Mathematical Analysis and Applications 484(2) (2020), pp. 123-130.

\bibitem{Mawardi-Ryuichi} B. Mawardi, and A. Ryuichi, A Simplified Proof of Uncertainty Principle for Quaternion Linear Canonical Transform. Abstract and Applied Analysis, Article ID 5874930, 11 pages. DOI: 10.1155/2016/5874930 (2016).

\bibitem{Mawardi-Ryuichi-1} B. Mawardi, and R. Ashino, Logarithmic uncertainty principle for quaternion linear canonical transform. Proceedings of the 2016 International Conference on Wavelet Analysis and Pattern Recognition, Jeju, South Korea, 10-13 July, 6 pages (2016).

\bibitem{Mawardi-Ryuichi-2} B. Mawardi, and R. Ashino, A Variation on Uncertainty Principle and Logarithmic Uncertainty Principle for Continuous Quaternion Wavelet Transforms. Abstract and Applied Analysis, Article ID 3795120, 11 pages, https://doi.org/10.1155/2017/3795120 (2017).

\bibitem{Mawardi-Hitzer-1} B. Mawardi, and E. Hitzer, Clifford Algebra $Cl(3,0)$-valued Wavelets and Uncertainty Inequality for Clifford Gabor Wavelet Transformation, Preprints of Meeting of the Japan Society for Industrial and Applied Mathematics, ISSN: 1345-3378, Tsukuba University, 16-18 Sep. 2006, Tsukuba, Japan, pp. 64-65 (2006).

\bibitem{Mawardi-Hitzer-2} B. Mawardi, and E. Hitzer, Clifford algebra $Cl(3,0)$-valued wavelet transformation, Clifford wavelet uncertainty inequality and Clifford Gabor wavelets. International Journal of Wavelets, Multiresolution and Information Processing 5(6), 997--1019 (2007).

\bibitem{Mawardi-Hitzer-3} B. Mawardi, and E. Hitzer, Clifford Fourier Transformation and Uncertainty Principle for the Clifford Geometric Algebra $Cl_{3,0}$. Adv. appl. Clifford alg. 16, 41--61. DOI 10.1007/s00006-006-0003-x (2006).

\bibitem{Mawardi-Hitzer-4} B. Mawardi, and E. Hitzer, Clifford Fourier Transform on Multivector Fields and Uncertainty Principles for Dimensions $n=2(mod4)$ and $n=3(mod4)$. Adv. appl. Clifford alg. 18, 715--736. DOI: 10.1007/s00006-008-0098-3 (2008).

\bibitem{Mawardi-Hitzer-Hayashi-Ashino} B. Mawardi, E. Hitzer, A. Hayashi, and R. Ashino, An Uncertainty Principle for Quaternion Fourier Transform, Computer \& Mathematics with Applications 56, 2398-2410 (2008).

\bibitem{Nagata} K. Nagata, and T. Nakamura, Violation of Heisenberg's Uncertainty Principle. Open Access Library Journal, 2, e1797. http://dx.doi.org/10.4236/oalib.1101797 (2015)

\bibitem{Rachdi-Meherzi} L. T. Rachdi, and F. Meherzi, Continuous Wavelet Transform and Uncertainty Principle Related to the Spherical Mean Operator. Mediterranean Journal of Mathematics, 14(1). doi:10.1007/s00009-016-0834-1 (2016)

\bibitem{Rachdi-Herch} L. T. Rachdi, and H. Herch, Uncertainty principles for continuous wavelet transforms related to the Riemann–Liouville operator. Ricerche Di Matematica, 66(2), 553–578. doi:10.1007/s11587-017-0320-5 (2017)

\bibitem{Sen} D. Sen, The uncertainty relations in quantum mechanics. Current Science 107(2), 203--218 (2018).

\bibitem{Stabnikov} P. A. Stabnikov, Geometric Interpretation of the Uncertainty Principle. Natural Science 11(5), 146-148 (2019).

\bibitem{Yangetal1} Y. Yang, P. Dang, and T. Qian, Stronger uncertainty principles for hypercomplex signals, Complex Variables and Elliptic Equations, 60(12), 1696--1711, DOI: 10.1080/17476933.2015.1041938 (2015).

\bibitem{Yang-Kou} Y. Yang, and K. I. Kou, Uncertainty principles for hypercomplex signals in the linear canonical transformdomains. Signal Processing 95, 67–75 (2014).

\bibitem{Zou} C. Zou, and K. I. Kou, Hypercomplex Signal Energy Concentration in the Spatial and Quaternionic Linear Canonical Frequency Domains, arXiv preprint arXiv:1609.00890 (2016)
\end{thebibliography}
\end{document}